\documentclass[letterpaper, onecolumn, 10 pt, conference]{ieeeconf}
\IEEEoverridecommandlockouts
\usepackage{amssymb, amsfonts, mathtools, xargs, tensor, units, cite, stmaryrd, mathrsfs, algpseudocode, algorithm, graphicx, xcolor}
\usepackage[unicode=true, colorlinks=false, hidelinks]{hyperref}
\usepackage{pgfplots}

\usepackage{setspace}
\newtheorem{assumption}{Assumption}
\newtheorem{definition}{Definition}
\newtheorem{remark}{Remark}
\newtheorem{lemma}{Lemma}
\newtheorem{theorem}{Theorem}

\DeclareMathOperator{\argmax}{argmax}
\DeclareMathOperator{\eig}{eig}

\DeclareMathOperator{\diag}{diag}

\allowdisplaybreaks

\title{\LARGE \bf State and Parameter Estimation for Affine Nonlinear Systems}
\author{Tochukwu Elijah Ogri$^{1}$ \and Zachary I. Bell$^{2}$ \and Rushikesh Kamalapurkar$^{1}$
\thanks{This research was supported in part by the Office of Naval Research under award number N00014-21-1-2481 and the Air Force Research Laboratories under contract number FA8651-19-2-0009. Any opinions, findings, or recommendations in this article are those of the author(s), and do not necessarily reflect the views of the sponsoring agencies.}%
\thanks{$^{1}$ School of Mechanical and Aerospace Engineering, Oklahoma State University, email: {\tt\footnotesize \{tochukwu.ogri, rushikesh.kamalapurkar\} @okstate.edu}.}%
\thanks{$^{2}$ Air Force Research Laboratories, Florida, USA, email: {
\tt \footnotesize zachary.bell.10@us.af.mil.}}}
\begin{document}
\maketitle
\thispagestyle{empty}
\pagestyle{empty}

\begin{abstract} 
This paper proposes a new approach to online state and parameter estimation for nonlinear systems that address limitations in traditional adaptive control methods. Conventional methods apply to a narrow class of  nonlinear systems and rely on stringent excitation conditions, which can be challenging to obtain in practice. In contrast, the proposed approach uses multiplier matrices and a data-driven concurrent learning method to develop an adaptive observer for affine nonlinear systems. Through Lyapunov-based analysis, the technique is proven to guarantee locally uniformly ultimately bounded stable state estimates and ultimately bounded parameter estimation errors. Additionally, under certain excitation conditions, the parameter estimation error is guaranteed to converge to a neighborhood of the origin. Simulation results confirm the effectiveness of the approach, even under mild excitation conditions, highlighting its potential applicability in real-world systems.
\end{abstract}

\section{Introduction}	
\label{section:introduction}
In many real-world systems, the limited availability of sensor information and unknown model parameters make effective control of the system difficult if not impossible. While adaptive control methods typically rely on measurement of the entire state to generate parameter estimates, techniques that can simultaneously estimate the system state are also available for specific classes of dynamical systems \cite{SCC.Kamalapurkar2017a, SCC.Kamalapurkar2017,SCC.katiyar.Bhasin.ea2022}.

Extended Luenberger nonlinear observer techniques have been extensively researched in the field of control theory for state estimation in nonlinear systems \cite{SCC.Accikmecse.Corless.ea2008, SCC.Wang.Bevly.ea2014,SCC.Karami.Farhad.ea2021,SCC.accikmecse.ea2011}. These techniques have been applied to a wide range of systems, including robotics, chemical processes, and power systems.
The nonlinear observer design techniques in \cite{SCC.Arcak.Kokotovic,SCC.Rajamani.Zemouche.ea2020,SCC.Karami.Farhad.ea2021, SCC.Wang.Bevly.ea2014,SCC.accikmecse.ea2011,Ogri.Kamalapurkar.easubmitted} utilize incremental multiplier matrices to characterize the nonlinearities in the system dynamics. The observer gain matrices are then obtained by solving linear matrix inequalities \cite{SCC.Wang.Bevly.ea2014, SCC.accikmecse.ea2011,Ogri.Kamalapurkar.easubmitted} using semi-definite programming. The convergence properties of these state observers can be leveraged to generate precise state estimates for parameter estimation, even when only partial state measurements are available from the output.

Parameter estimation methods that rely on persistent excitation (PE) \cite{SCC.Ioannou.Sun1996,SCC.Anderson1977,SCC.Green.Moore1986} and finite excitation \cite{SCC.Chowdhary.Johnson2011,SCC.Chowdhary.Yucelen.ea2013,SCC.Kersting.Buss2014,SCC.Chowdhary.Muehlegg.ea2013} have also been studied extensively in the literature for systems where all state variables can be measured.

Recent research efforts have focused on the development of adaptive observers that can simultaneously estimate the state and parameters of nonlinear systems \cite{SCC.Creveling.Gill.ea2008, SCC.Togneri.Deng2003, SCC.Chong.Nesic.ea2014, SCC.Kamalapurkar2017a, SCC.Kamalapurkar2017, SCC.katiyar.Bhasin.ea2022, SCC.Pyrkin.Anton.ea2019, arXivSCC.Liu.Buss.ea2022}. However, most of these techniques are restricted to a specific class of nonlinear systems and rely on assumptions that may be difficult to satisfy in practice, such as stringent PE conditions\cite{SCC.Kamalapurkar2017a, SCC.Kamalapurkar2017, arXivSCC.Liu.Buss.ea2022, SCC.katiyar.Bhasin.ea2022, SCC.Pyrkin.Anton.ea2019}, or employ assumptions detailed below which reduce their applicability to most real-world systems.

The design of extended Luenberger nonlinear observers is challenging due to the need for the computation of bounds on the Jacobian matrices of unknown vector fields that model the system. Methods such as \cite{SCC.Quintana.Bernal.ea2021} and \cite{SCC.Guerra.Bernal.ea2018} offer solutions to challenges in the calculation of such Jacobian bounds and suitable multiplier matrices. However,  due to their separation of measurable and unmeasurable signals and reliance on convex optimization techniques to formulate an explicit matrix polynomial form of the gradient, methods such as \cite{SCC.Quintana.Bernal.ea2021} and \cite{SCC.Guerra.Bernal.ea2018} are difficult to apply for simultaneous state and parameter estimation in nonlinear systems.
  
Methods such as those developed in \cite{SCC.Kamalapurkar2017a} and \cite{SCC.Kamalapurkar2017}, while effective, are restricted to a class of dynamical systems that are of the Brunovsky canonical form. Similarly, adaptive observers that use dynamic regressor extension and mixing (DREM) rely on the existence of a cascade form via a coordinate change for which a linear regression relation exists between the system states and unknown parameters \cite{SCC.Pyrkin.Anton.ea2019}. Unlike the existing simultaneous state and parameter estimation methods described above, which are limited to narrow classes of systems such as systems in Brunovsky form or cascade form, this paper presents a novel method that achieves simultaneous state and parameter estimation for a broader class of affine nonlinear systems. 

This paper develops a joint state and parameter estimation scheme by leveraging the advantages of the multiplier matrix approach for observer design, which has been proven to yield asymptotic convergence of state estimation errors \cite{SCC.Wang.Bevly.ea2014}, and builds upon concurrent learning (CL) frameworks \cite{SCC.Kamalapurkar2017a,SCC.Kamalapurkar2017}, which utilizes recorded data (stored in what is commonly called a history stack) to estimate parameters with high accuracy. In contrast with methods proposed in results such as \cite{SCC.Yang.Liu.Huang.ea2013,SCC.Yang.Liu.ea2014,SCC.Huang.Jiang2015}, the developed method does not require any restrictions on the form and rank of the $C$ measurement matrix or impose observability conditions. 

The rest of the paper is organized as follows: Section ~\ref{section:problemFormulation} contains the problem formulation, Section ~\ref{section:stateEstimator} introduces the  state estimator/observer, Section ~\ref{section:parameterEstimator} presents the parameter estimator design, Section ~\ref{section:stabilityAnalysis} contains stability analysis of the developed method, Section ~\ref{section:simulation} contains simulation results and Section ~\ref{section:conclusion} concludes the paper.



\section{Problem Formulation}	
\label{section:problemFormulation}
This paper considers nonlinear dynamical systems of the form 
\begin{equation}
\dot{x} = Y(x)\theta+f_{0}(x) + g(x)u , \quad y = Cx, \label{eq:dynamics_x}
\end{equation}
where $x \in \mathbb{R}^{n}$ and , $u \in\mathbb{R}^{m}$ denotes the system state and the control input respectively, $\theta \in \Theta \subset \mathbb{R}^{p}$ is a vector of unknown parameters, $C \in \mathbb{R}^{q \times n}$ is the output matrix, and $y \in \mathbb{R}^{q}$ is the measured output. The functions $Y: \mathbb{R}^{n} \rightarrow \mathbb{R}^{n \times p}$, $f_{0}: \mathbb{R}^{n} \rightarrow \mathbb{R}^{n}$, and $g: \mathbb{R}^{n} \rightarrow \mathbb{R}^{n \times m}$, denote the regressor, the drift, and the control effectiveness, respectively. 

The objective is to design a state observer that estimates the state $x$ online using the output $y$ and the input $u$ and develop a parameter estimation scheme, which uses memory from recorded data to provide parameter estimates denoted as $\hat{\theta}$. Consistent with the literature on state and parameter estimation, it is assumed that the control signal $u$ and the system state $x$ are bounded.

In order to facilitate the development and analysis of the method presented in this paper, the following assumption is necessary. 
\begin{assumption}\label{ass:jacobianbounds}
    The functions $Y$, $f_0$ and  $g$ are known, their derivatives exist on the compact set $\mathcal{C}\subset\mathbb{R}^n$ and satisfy the element-wise bounds
    \begin{align}
    (K_{y_{1}})_{j,k}  \leq \left(\frac{\mathrm{d}(Y(x))_{j,i}}{\mathrm{d}(x)_k}\right)\theta_{i} \leq (K_{y_{2}})_{j,k}, \label{eq:yBound}\\
	(K_{f_{1}})_{j,k}  \leq \frac{\mathrm{d}(f_{0}(x))_j}{\mathrm{d}(x)_k} \leq (K_{f_{2}})_{j,k}, \\
	(K_{g_{1}})_{j,k}  \leq \left(\frac{\mathrm{d}(g(x))_{j,l}}{\mathrm{d}(x)_k}\right)u_{l} \leq (K_{g_{2}})_{j,k},
\end{align}
for all $x\in\mathcal{C}$, $u\in \mathcal{U}$, $\theta\in\Theta$, $i=1,\ldots,p$,  $j,k = 1,\ldots,n$,  and $l=1,\ldots,m$, where $(\cdot)_{i}$, $(\cdot)_{j}$, $(\cdot)_{i,k}$, and $(\cdot)_{j,k}$ denote the element of the array $(\cdot)$ at the index indicated by the subscript.
\end{assumption}
\begin{remark}
The conditions stated in Assumption~\ref{ass:jacobianbounds} are commonly required in several observer design schemes (see, e.g., \cite{SCC.Zemouche.Bara.ea2005,SCC.Wang.Bevly.ea2014,SCC.Wang.Yang.ea2017,SCC.Rajamani.Zemouche.ea2020}).
\end{remark}

In the following section, sufficient conditions involving multiplier matrices that characterize the affine system will be presented, along with the design of the state estimator.

\section{State Estimator}\label{section:stateEstimator}
This section presents the development of a state estimator that generates estimates of $x$ by employing an extended Luenberger-like observer. To facilitate the observer design, the nonlinear dynamics described in (\ref{eq:dynamics_x}) is expressed as
\begin{equation}
    \dot{x} =	Ax+ F_{\theta}(x, \theta)+F_{1}(x) + G_{u}(x,u) , \quad y = Cx,\label{aug_dynamics_x}
\end{equation}
where
$A = (K_{y_{1}} + K_{f_{1}}+K_{g_{1}})$, $F_{\theta}(x, \theta) = -K_{y_{1}}x+Y(x)\theta$, $F_{1}(x) = -K_{f_{1}}x+f_{0}(x)$,  and $G_{u}
(x,u) = -K_{g_{1}}x + \sum_{i=1}^{N} g_i(x)(u)_i.$

 Assumption~\ref{ass:jacobianbounds} implies that the derivatives of $F_{\theta}$, $F_{1}$ and $G_{u}$ satisfy the element-wise inequalities
\begin{gather}
	0 \leq \frac{\mathrm{d}(F_{\theta}(x))_j}{\mathrm{d}(x)_k} \leq (K_{y_{2}})_{j,k}-(K_{y_{1}})_{j,k},\label{aug_jac_y} \\ 
	0 \leq \frac{\mathrm{d}(F_{1}(x))_j}{\mathrm{d}(x)_k} \leq (K_{f_{2}})_{j,k}-(K_{f_{1}})_{j,k}, \label{aug_jac_f} \\ 
	0 \leq \frac{\mathrm{d}(G_{u}(x,u))_{j}}{\mathrm{d}(x)_k}\leq (K_{g_{2}})_{j,k}-(K_{g_{1}})_{j,k},\label{aug_jac_g}
\end{gather}
Using the derivative bounds, a state estimator with four correction terms is designed as
\begin{equation}
    	\dot{\hat{x}} = A\hat{x}+F_{\theta}[\hat{x}+l_{1}(y-C\hat{x}),\hat{\theta}] + F_{1}[\hat{x}+l_{2}(y-C\hat{x})]+ G_{u}[\hat{x}+l_{3}(y-C\hat{x}), u]+L(y-C\hat{x}),\label{eq:observerdynamics_x}
\end{equation}
where $\hat{x}  \in \mathbb{R}^{n}$ is the
estimate of $x$, $l_{1} \in \mathbb{R}^{n \times q}$, $l_{2} \in \mathbb{R}^{n \times q}$, $l_{3} \in \mathbb{R}^{n \times q}$ and $L \in \mathbb{R}^{n \times q}$ are observer gains, $l_{1}(y-C\hat{x})$, $l_{2}(y-C\hat{x})$ and $l_{3}(y-C\hat{x})$ are nonlinear injection terms and  $L(y-C\hat{x})$ is a linear correction term. With the state estimation error defined as
\begin{equation}
	\tilde{x}  \coloneqq x -\hat{x}\label{eq:deWstimation_error_x},
\end{equation}
estimation error dynamics is given by
\begin{multline}\label{aug_error}
	\dot{\tilde{x}} = (A-LC)\tilde{x} + F_{\theta}(x, \hat\theta)  +F_{1}(x)+G_{u}(x,u)-F_{\theta}[\hat{x}+l_{1}(y-C\hat{x}),\hat{\theta}]-F_{1}[\hat{x}+l_{2}(y-C\hat{x})]\\-G_{u}[\hat{x}+l_{3}\left(y-C\hat{x}\right),u] + F_{\theta}(x, \tilde{\theta}).
\end{multline}
Let the parameter estimation error be defined as
\begin{equation}\label{eq:param_error}
    \tilde{\theta} \coloneqq \theta - \hat{\theta}.
\end{equation}
To facilitate the design of the state estimator, the following assumption is made about the set $\Theta$ which contains $\theta$.
\begin{assumption}\label{ass:ThetaSet}
     There exist a known constant $\overline{\theta} \in \mathbb{R}$ such that $\|\theta\| \leq \overline{\theta}$.
\end{assumption}
\begin{remark}
    Assumption~\ref{ass:ThetaSet} is used to implement a parameter projection algorithm that maintains apriori boundedness of the parameter estimates $\hat{\theta}(t)$ and ensures it stays within a bounded convex set $\Theta \subset \mathbb{R}^{p} \coloneqq \{\hat{\theta} \mid h(\hat{\theta}) \leq 0\}, \text{where } h(\hat{\theta}) \coloneqq \hat{\theta}^\mathsf{T}\hat{\theta}-\overline{\theta}^{2}$ and $\nabla_{\hat{\theta}}h(\hat{\theta}) \coloneqq 2\hat{\theta}$ (cf. \cite[Example~4.4.2]{SCC.Ioannou.Sun1996}).
\end{remark}
The difference functions between the uncertain system components and their estimates can then be characterized using $\psi_{y}: \mathbb{R}_{\geq 0} \times \mathcal{D} \times \Theta \to \mathbb{R}^{n}$, $\psi_{f}: \mathbb{R}_{\geq 0} \times \mathcal{D} \to \mathbb{R}^{n}$, and $\psi_{g}: \mathbb{R}_{\geq 0} \times \mathcal{D} \to \mathbb{R}^{n}$, defined as
\begin{gather}
    \psi_{y}(t,\tilde{x},\hat{\theta}) \coloneqq  F_{\theta}(x, \hat\theta)-F_{\theta}[\hat{x}+l_{1}(y-C\hat{x}),\hat{\theta}],\\
    \psi_{f}(t,\tilde{x}) \coloneqq   F_{1}(x)-F_{1}[\hat{x}+l_{2}(y-C\hat{x})], \text{ and }\\
    \psi_{g}(t,\tilde{x})  \coloneqq  G_{u}(x,u)-G_{u}[\hat{x}+l_{3}(y-C\hat{x}),u].
\end{gather}
The observer error dynamics in (\ref{aug_error}) can then be expressed as
\begin{equation}\label{aug_error2}
     \dot{\tilde{x}} = \left(A-LC\right)\tilde{x} + \psi_{y}(t,\tilde{x}, \hat{\theta}) + \psi_{f}(t,\tilde{x}) + \psi_{g}(t,\tilde{x}) + F_{\theta}(x, \tilde\theta).
\end{equation}

According to the differential mean value theorem (DMVT) \cite[Theorem 2.1]{SCC.Zemouche.ea2005}, provided Assumption~\ref{ass:jacobianbounds} holds and Assumption~\ref{ass:ThetaSet}, the difference functions $\psi_{y}$, $\psi_{f}$, and $\psi_{g}$ are guaranteed to be bounded as
\begin{equation}\label{eq:y_ineq}
 \bar{K}_{y_{1}}(\mathbb{I}_{n}-l_{1}C)\tilde{x}\leq \psi_{y}(t,\tilde{x}, \hat{\theta})  \leq \bar{K}_{y_{2}}(\mathbb{I}_{n}-l_{1}C)\tilde{x},
\end{equation}
\begin{equation}\label{eq:f_ineq}
 \bar{K}_{f_{1}}(\mathbb{I}_{n}-l_{2}C)\tilde{x}\leq \psi_{f}(t,\tilde{x}) \leq \bar{K}_{f_{2}}(\mathbb{I}_{n}-l_{2}C)\tilde{x},
\end{equation}
and
\begin{equation}\label{eq:g_ineq}
\bar{K}_{g_{1}}(\mathbb{I}_{n}-l_{3}C)\tilde{x}\leq \psi_{g}(t,\tilde{x}) \leq \bar{K}_{g_{2}}(\mathbb{I}_{n}-l_{3}C)\tilde{x}.
\end{equation}
where $\bar{K}_{y_{1}} = 0_{n \times n}$, $\bar{K}_{y_{2}} = K_{y_{2}}-K_{y_{1}}$, $\bar{K}_{f_{1}} = 0_{n \times n}$, $\bar{K}_{f_{2}} = K_{f_{2}}-K_{f_{1}}$, $\bar{K}_{g_{1}} = 0_{n \times n}$ and $\bar{K}_{g_{2}} = K_{g_{2}}-K_{g_{1}}$.

To establish the stability of the state estimation error dynamics, it suffices to rely on the sector information provided by the compact set $\mathcal{C}$, which is defined by the Jacobian bounds presented in (\ref{aug_jac_y}), (\ref{aug_jac_f}), and (\ref{aug_jac_g}) and constraints $\psi_{y}$, $\psi_{f}$, and $\psi_{g}$. Specifically, the inequalities in (\ref{eq:y_ineq}), (\ref{eq:f_ineq}), and (\ref{eq:g_ineq}) can be used to obtain the following bounds
\begin{equation}\label{quad_y}
 \left[ \psi_{y}(t,\tilde{x}, \hat{\theta}) \right]^{\mathsf{T}}[ \psi_{y}(t,\tilde{x}, \hat{\theta}) -\bar{K}_{y_{2}}(\mathbb{I}_{n}-l_{1}C)\tilde{x}]\leq 0,
\end{equation}
\begin{equation}\label{quad_f}
    \left[ \psi_{f}(t,\tilde{x})\right]^{\mathsf{T}}[ \psi_{f}(t,\tilde{x})-\bar{K}_{f_{2}}(\mathbb{I}_{n}-l_{2}C)\tilde{x}]\leq 0,
\end{equation}
and
\begin{equation}\label{quad_g}
 \left[ \psi_{g}(t,\tilde{x})\right]^{\mathsf{T}}[ \psi_{g}(t,\tilde{x})-\bar{K}_{g_{2}}(\mathbb{I}_{n}-l_{3}C)\tilde{x}]\leq 0.
\end{equation}
Converting the inequalities in (\ref{quad_y}), (\ref{quad_f}), and (\ref{quad_g}) into their quadratic forms yields
\begin{equation}\label{matrixCondB}
\begin{bmatrix} 
\tilde{x} \\ \psi_{y}
\end{bmatrix}^{\mathsf{T}}\begin{bmatrix}
\mathbb{I}_{n}-l_{1}C & 0\\ 0 &\mathbb{I}_{n}
\end{bmatrix}^{\mathsf{T}}J_{y}\begin{bmatrix}
\mathbb{I}_{n}-l_{1}C & 0\\ 0 &\mathbb{I}_{n}
\end{bmatrix}\begin{bmatrix}
\tilde{x} \\ \psi_{y}
\end{bmatrix} \leq 0, 
\end{equation}

\begin{equation}\label{matrixCondH}
\begin{bmatrix} 
\tilde{x} \\\psi_{f}
\end{bmatrix}^{\mathsf{T}}\begin{bmatrix}
\mathbb{I}_{n}-l_{2}C & 0\\ 0 &\mathbb{I}_{n}
\end{bmatrix}^{\mathsf{T}}J_{f}\begin{bmatrix}
\mathbb{I}_{n}-l_{2}C & 0\\ 0 &\mathbb{I}_{n}
\end{bmatrix}\begin{bmatrix}
\tilde{x} \\\psi_{f}
\end{bmatrix} \leq 0, 
\end{equation}
and
\begin{equation}\label{matrixCondK}
\begin{bmatrix}
\tilde{x} \\ \psi_{g}
\end{bmatrix}^{\mathsf{T}}\begin{bmatrix}
\mathbb{I}_{n}-l_{3}C & 0\\ 0 &\mathbb{I}_{n}
\end{bmatrix}^{\mathsf{T}}J_{g}\begin{bmatrix}
\mathbb{I}_{n}-l_{3}C & 0\\ 0 &\mathbb{I}_{n}
\end{bmatrix}\begin{bmatrix}
\tilde{x} \\ \psi_{g}
\end{bmatrix} \leq 0, 
\end{equation}
with their corresponding multiplier matrices designed as
\begin{equation}\label{J_y}
J_{y} = \begin{bmatrix}
0 & -\frac{K_{y_{2}}^{\mathsf{T}}-K_{y_{1}}^{\mathsf{T}}}{2}\\ -\frac{K_{y_{2}}-K_{y_{1}}}{2} &\mathbb{I}_{n}
\end{bmatrix},
\end{equation}

\begin{equation}\label{J_f}
J_{f} = \begin{bmatrix}
0 & -\frac{K_{f_{2}}^{\mathsf{T}}-K_{f_{1}}^{\mathsf{T}}}{2}\\ -\frac{K_{f_{2}}-K_{f_{1}}}{2} &\mathbb{I}_{n}
\end{bmatrix},
\end{equation}
and
\begin{equation}\label{J_g}
J_{g} = \begin{bmatrix}
0& -\frac{K_{g_2}^{\mathsf{T}}-K_{g_1}^{\mathsf{T}}}{2}\\ -\frac{K_{g_2}-K_{g_1}}{2} &\mathbb{I}_{n}
\end{bmatrix}. 
\end{equation}

A Lyapunov-based analysis that uses the above inequalities to establish boundedness of the state estimation error for all $t \in \mathbb{R}_{\geq 0}$ is presented in Section~\ref{section:stabilityAnalysis}.

\section{Parameter Estimator Design}
\label{section:parameterEstimator}

In this section, an output feedback concurrent learning update law is developed to estimate the unknown parameters.
The parameter estimator relies on the fact that the difference between the state estimates at time $t$ and time $t-T$ can be expressed as an affine function of the parameters $\theta$ and a residual that reduces with reducing state estimation errors.

\begin{lemma}\label{lem:ErrorTermformulation}

If $x$, $\hat{x} \in \mathcal{C}$ and if Assumption~\ref{ass:jacobianbounds} holds, then a given time delay $T\geq 0$ and for all $t\geq T$, the state estimates satisfy $\hat{x}(t) -\hat{x}(t-T) = \hat{\mathcal{Y}}(t) \theta+\hat{\mathcal{G}}_{fu}(t) + \mathcal{E}(t)$, where $\hat{\mathcal{Y}}(t)\coloneqq \int_{t-T}^{t} Y(\hat{x}(\tau))d\tau$, $\hat{\mathcal{G}}_{fu}(t) \coloneqq \int_{t-T}^{t} f_{0}(\hat{x}(\tau))+ g(\hat{x}(\tau))u(\tau)d\tau$, and $\mathcal{E}(t) =  O\left(\sup_{\sigma\in[t-T,t]}\|\tilde{x}(\sigma)\|\right)$.

\end{lemma}
\begin{proof}
    Integrating the dynamics in (\ref{eq:dynamics_x}) yields
    \begin{equation}\label{eq:dyanmicsIntegral}
        x(t)-x(t-T) = \int_{t-T}^{t} Y(x(\tau))\theta + f_{0}(x(\tau))+ g(x(\tau))u(\tau)d\tau
    \end{equation}
    By adding and subtracting $x(t)$ and $x(t-T)$, the difference $\hat{x}(t)-\hat{x}(t-T)$ can be expressed as
    \begin{equation}
        \hat{x}(t)-\hat{x}(t-T) = -\tilde{x}(t)+\tilde{x}(t-T)+x(t)-x(t-T)
    \end{equation}
    Substituting from (\ref{eq:dyanmicsIntegral}), adding and subtracting the integral $ \int_{t-T}^{t} Y(\hat{x}(\tau))\theta+f_{0}(\hat{x}(\tau))+ g(\hat{x}(\tau))u(\tau)d\tau$, and simplifying yields 
    \begin{multline}
        \hat{x}(t)-\hat{x}(t-T) =\int_{t-T}^{t} Y(\hat{x}(\tau))\theta d\tau + \int_{t-T}^{t} f_{0}(\hat{x}(\tau))+ g(\hat{x}(\tau))u(\tau)d\tau -\tilde{x}(t) + \tilde{x}(t-T) + \int_{t-T}^{t} \tilde{Y} (x(\tau),\hat{x}(\tau))\theta d\tau \\+ \int_{t-T}^{t} \tilde{f}_{0}(x(\tau),\hat{x}(\tau))+\tilde{g}(x(\tau),\hat{x}(\tau), u(\tau))d\tau
    \end{multline}
    where $\tilde{Y}(x(\tau),\hat{x}(\tau)) \coloneqq Y(x(\tau))-Y(\hat{x}(\tau))$, $\tilde{f}_{0}(x(\tau),\hat{x}(\tau)) \coloneqq f_{0}(x(\tau))-f_{0}(\hat{x}(\tau))$, and $\tilde{g}(x(\tau),\hat{x}(\tau),u(\tau)) \coloneqq g(x(\tau))u(\tau)-g(\hat{x}(\tau))u(\tau)$. If $x,\hat{x} \in \mathcal{C}$  and Assumption~\ref{ass:jacobianbounds} holds, then the DMVT can be invoked to obtain 
    \begin{equation}
        \hat{x}(t)-\hat{x}(0) = \hat{\mathcal{Y}}(t)\theta +\hat{\mathcal{G}}_{fu}(t) + \mathcal{E}(t)
    \end{equation}
    where the residual term satisfies $\mathcal{E}(t) =  O\left(\sup_{\sigma\in[t-T,t]}\|\tilde{x}(\sigma)\|\right)$.
\end{proof}
{Lemma~\ref{lem:ErrorTermformulation} implies that the parameter estimation error at any time $\tau$ can be expressed as $\hat{\mathcal{Y}}(t) \tilde{\theta}(\tau) = \hat{x}(t) -\hat{x}(t-T) - \hat{\mathcal{G}}_{fu}(t)-\hat{\mathcal{Y}}(t) \hat{\theta}(\tau) - \mathcal{E}(t) $, which motivates the update law 
{\medmuskip=0mu\thinmuskip=0mu\thickmuskip=0mu\begin{equation}\label{eq:parameterUpdate}
 \dot{\hat{\theta}} =  \begin{cases} 
      k_{\theta}\Gamma\phi, & \text{if }\hat{\theta}^\mathsf{T}\hat{\theta} \ < \ \overline{\theta}^{2}\text{or if } \\& \hat{\theta}^\mathsf{T}\hat{\theta}  = \overline{\theta}^{2} \text{and }(k_{\theta}\Gamma\phi)^\mathsf{T}\hat{\theta} \leq 0\\
       \left(\mathbb{I}_{p} - \frac{\Gamma\hat{\theta}\hat{\theta}^\mathsf{T}}{\hat{\theta}^\mathsf{T}\Gamma\hat{\theta}}\right)k_{\theta}\Gamma\phi, & \text{otherwise}
   \end{cases}\end{equation}}
where $\phi \coloneqq \sum_{i=1}^{N} \left(\frac{\hat{\mathcal{Y}}(t_{i})}{1+\kappa\|\hat{\mathcal{Y}}(t_{i})\|^{2}}\right)^{\mathsf{T}}\big(\hat{x}(t_{i})-\hat{x}(t_{i}-T) -\hat{\mathcal{G}}_{fu}(t_{i})-\hat{\mathcal{Y}}(t_{i})\hat{\theta})$, $\kappa \in \mathbb{R}_{>0}$ is the normalization gain, and $k_{\theta} \in \mathbb{R}_{>0}$ is the CL gain. The matrix $\Gamma \in \mathbb{R}^{p \times p}$ is the least-squares gain matrix updated as
{\medmuskip=0mu\thinmuskip=0mu\thickmuskip=0mu\begin{equation}\label{eq:gammaUpdateLaw}
    \dot{\Gamma} = \begin{cases} 
       \beta_{1}\Gamma - k_{\theta}\Gamma\frac{{\hat{\mathcal{Y}}(t_{i})^{\mathsf{T}}\hat{\mathcal{Y}}(t_{i})}}{1+\kappa\|\hat{\mathcal{Y}}(t_{i})\|^{2}} \Gamma, & \text{if }\hat{\theta}^\mathsf{T}\hat{\theta} \ < \ \overline{\theta}^{2} \text{or if }\\& \hat{\theta}^\mathsf{T}\hat{\theta}  = \overline{\theta}^{2} \text{and }(k_{\theta}\Gamma\phi)^\mathsf{T}\hat{\theta} \leq 0\\
       0 & \text{otherwise}
   \end{cases}  
\end{equation}}
where $\beta_{1} \in \mathbb{R}_{>0}$ is a constant adaptation gain. The update law relies on the time delay $T$, and a history stack $\mathcal{H}$. The history stack represents a set of piecewise constant functions that can be expressed as follows:
 \begin{gather}
     \hat{\mathscr{X}} = \begin{bmatrix}
                   \hat{x}(t_{1})-\hat{x}(t_{1}-T)\\
                   \vdots\\
                  \hat{x}(t_{N})-\hat{x}(t_{N}-T)
                    \end{bmatrix} \in \mathbb{R}^{nN},
                    \hat{\mathscr{Y}} = \begin{bmatrix}
                    \hat{\mathcal{Y}}(t_{1})\\
                   \vdots\\
                   \hat{\mathcal{Y}}(t_{N})
                    \end{bmatrix} \\ \in \mathbb{R}^{nN\times p},
                     \hat{\mathscr{G}}_{fu} = \begin{bmatrix}
                    \hat{\mathcal{G}}_{fu}(t_{1})\\
                   \vdots\\
                   \hat{\mathcal{G}}_{fu}(t_{N})
                    \end{bmatrix} \in \mathbb{R}^{nN}
 \end{gather}
The integral terms in the update law can be calculated as
\begin{gather}
    \hat{\mathcal{Y}}(t) = I_{Y}(t)-I_y(t-T), \ \text{and} \\ \mathcal{G}_{fu}(t) = \hat{I}_{fg_{u}}(t)-\hat{I}_{fg_{u}}(t-T),
\end{gather}
where $\hat{I}_{Y}(t) = \int_{0}^{t} Y(\hat{x})$ and $\hat{I}_{fg_{u}}(t) = \int_{0}^{t} f_{0}(\hat{x}) + g(\hat{x})u$, are computed by solving
\begin{equation}\label{regressorUpdate}
     \dot{\hat{I}}_{Y} = Y(\hat{x})
\end{equation}
and
\begin{equation}\label{ifguUpdate}
    \dot{\hat{I}}_{fg_{u}} = f_{0}(\hat{x}) + g(\hat{x})u
\end{equation}
starting from the initial conditions $I_{Y,0} = 0_{n\times p}$ and $I_{fg_{u},0} = 0_{n\times 1}$.

Using Lemma \ref{lem:ErrorTermformulation}, the parameter estimation error dynamics can be expressed as 
\begin{equation}\label{eq:parameterErrorDyn}
    \dot{\tilde{\theta}} = -k_{\theta}\Gamma \sum_{i=1}^{N}\frac{{\hat{\mathcal{Y}}(t_{i})^{\mathsf{T}}\hat{\mathcal{Y}}(t_{i})}}{1+\kappa\|\hat{\mathcal{Y}}(t_{i})\|^{2}}\tilde{\theta}-k_{\theta}\Gamma \sum_{i=1}^{N}\frac{{\hat{\mathcal{Y}}(t_{i})^\mathsf{T}}\mathcal{E}(t_{i})}{1+\kappa\|\hat{\mathcal{Y}}(t_{i})\|^{2}}.
\end{equation}
It is clear from \eqref{eq:parameterErrorDyn} that for the parameter estimation error to be bounded, the matrix $\sum_{i=1}^{N}\frac{{\hat{\mathcal{Y}}(t_{i})^{\mathsf{T}}\hat{\mathcal{Y}}(t_{i})}}{1+\kappa\|\hat{\mathcal{Y}}(t_{i})\|^{2}}$ needs to be positive definite. Such a positive definite matrix can be obtained if the trajectories are sufficiently informative and the data stored in the history stack $\mathcal{H}$ are recorded carefully. This requirement is formalized by the following assumption.
\begin{assumption}\label{ass:regressorRank}
    For a given $M \in \mathbb{N}$, there exist a set of time instances $\{t_{i}\}_{i=1}^{N}$ such that
    \begin{equation*}
        \lambda_{\min}\left(\sum_{i=1}^{N} \frac{{\hat{\mathcal{Y}}(t_{i})^{\mathsf{T}}\hat{\mathcal{Y}}(t_{i})}}{1+\kappa\|\hat{\mathcal{Y}}(t_{i})\|^{2}}\right) = \underline{c} > 0.
    \end{equation*}
\end{assumption}
In the following, a history stack that meets the eigenvalue condition in Assumption \ref{ass:regressorRank} is called \emph{full rank}.

Since the rate of convergence of the parameter estimation errors depends on the lower bound $\underline{c}$ on the minimum eigenvalue, a minimum eigenvalue maximization algorithm is utilized for the selection of the time instances $\{t_{i}\}_{i=1}^{N}$ where history stack $\mathcal{H}$ is updated (see, for example, \cite{SCC.Kamalapurkar2017}). The algorithm presented in Algorithm~\ref{algo:parameterEstimatorAlgo} replaces an existing data point $\left(\hat{x}_{i} - \hat{x}_{i-T}, \hat{\mathcal{Y}}_{i}, \hat{\mathcal{G}}_{fui}\right)$, with a new data point $\left(\hat{x}^{*} - \hat{x}^{*-}, \hat{\mathcal{Y}}^{*}, {\hat{\mathcal{G}}_{fu}}^{*}\right)$, for some $i \in {1, \hdots , N}$, where $\hat{x}^{*} - \hat{x}^{*-} \coloneqq \hat{x}(t)-\hat{x}(t-T)$, $\hat{\mathcal{Y}}^{*} \coloneqq \hat{\mathcal{Y}}(t)$, $\hat{\mathcal{G}}_{fu}^{*} \coloneqq \hat{\mathscr{G}}_{fu}(t)$, and $\hat{\mathcal{Y}}^{*} \coloneqq \hat{\mathcal{Y}}(t)$, only if the following condition holds:
\begin{equation}\label{eq:maximCond}
\lambda{\min}\left(\sum_{i \neq j} \sigma_{i}\hat{\mathcal{Y}}_{i}^\mathsf{T}\hat{\mathcal{Y}}_{i} + \sigma_{j}\hat{\mathcal{Y}}_{j}^\mathsf{T}\hat{\mathcal{Y}}_{j}\right) <  \frac{\lambda_{\min}\left(\sum_{i \neq j} \sigma_{i}\hat{\mathcal{Y}}_{i}^\mathsf{T}\hat{\mathcal{Y}}_{i} + \sigma^{*}\hat{\mathcal{Y}}^\mathsf{T}\hat{\mathcal{Y}}^{*}\right)}{\left(1 + \delta\right)}
\end{equation}
Here, $\lambda_{\min}(\cdot)$ denotes the minimum eigenvalue of a matrix, $\delta$ is a constant that can be adjusted, $\sigma_{i} \coloneqq \frac{1}{1+\kappa\|\hat{\mathcal{Y}}_{i}\|^{2}}$, $\sigma_{j} \coloneqq \frac{1}{1+\kappa\|\hat{\mathcal{Y}}_{j}\|^{2}}$, and $\sigma^{*} \coloneqq \frac{1}{1+\kappa\|\hat{\mathcal{Y}}^{*}\|^{2}}$.

The availability of accurate state estimates is required for precise parameter estimation. However, the initial history stack, recorded during transients, may contain inaccurate data, requiring a purge of the history stack once more accurate state estimates become available. In such cases, newer state estimates are preferred, subject to the conditions of Theorem~\ref{thm:stateErrorConvergence}. To ensure estimator stability while utilizing newer data, a greedy purging algorithm based on dwell time is employed. This algorithm uses two history stacks: a main stack denoted as $\mathcal{H}$ and a transient stack labeled $\mathcal{G}$. The transient stack is filled until a sufficient dwell time $\mathcal{T}$ has elapsed. Then, the main stack is purged, and the transient stack is copied into the main stack. This approach enables the use of newer, more accurate data while maintaining estimator stability.
\begin{algorithm}
     \caption{Algorithm for Event-based implementation of Concurrent learning Adaptive History Stack Observer. At each time instance $t$, $\tau_{1}$ stores the last time instance an event occurred, $\tau_{2}$ stores the last time instance $\mathcal{H}$ was purged, $\lambda$ stores the highest minimum eigenvalue encountered so far, $\mathcal{T}$ denotes the dwell time, $\lambda^{*}$ denotes some user selected eigenvalue threshold, $t^{*}$ denotes some user selected sampling rate and $\xi \in (0,1]$ is a threshold for purging.
 }
\begin{algorithmic}[1] 
 \Require $t_{f} \ \in \ \mathbb{R}_{\geq t_{0}}$, $t^{*} > 0$, $T \in \mathbb{R}_{\geq 0}$, $\lambda^{*} \geq 0$
\State $\hat{\mathscr{X}} \gets 0$,  $\hat{\mathscr{Y}} \gets 0$, $\hat{\mathscr{G}}_{fu} \gets 0$, $\tau_{1} \gets 0$, $\tau_{2} \gets 0$ \Comment{Global variables}
\State $\lambda \gets \min(\eig(\hat{\mathscr{Y}}^{\mathsf{T}}\hat{\mathscr{Y}}))$, $t_{0} \gets 0$, $\hat{x}_{0} \gets \hat{x}(t_{0})$, $\hat{\theta}_{0} = \hat{\theta}(t_{0})$
\While{$t_{0} < t_f$}
    \State integrate the DDEs in (\ref{eq:parameterUpdate}), (\ref{eq:gammaUpdateLaw}), (\ref{regressorUpdate}), and  (\ref{ifguUpdate}) over the interval, $[t_{0}, t_{f}]$
    \If{$(t-\tau_{1}) \geq t^{*}$}
        \If{$t \geq T$}
        \State stop integration, an event has occurred
        \State $j \gets \argmax_{i=1:N}\{\min\{\eig\left(\hat{\mathscr{Y}}^{\mathsf{T}}\hat{\mathscr{Y}}-\hat{\mathcal{Y}}_{i}^{\mathsf{T}}\hat{\mathcal{Y}}_{i}+\hat{\mathcal{Y}}^{\mathsf{T}}\hat{\mathcal{Y}}\right)\}\}$
            \If{
                 $\max_{i=1:N}\{\min\{\eig\left(\hat{\mathscr{Y}}^{\mathsf{T}}\hat{\mathscr{Y}}-\hat{\mathcal{Y}}_{i}^{\mathsf{T}}\hat{\mathcal{Y}}_{i}+\hat{\mathcal{Y}}^{\mathsf{T}}\hat{\mathcal{Y}}\right)\}\}-\lambda \geq \lambda^{*}$}
                 \State 
                 $\lambda \gets \max_{i=1:N}\{\min\{\eig\left(\hat{\mathscr{Y}}^{\mathsf{T}}\hat{\mathscr{Y}}-\hat{\mathcal{Y}}_{i}^{\mathsf{T}}\hat{\mathcal{Y}}_{i}+\hat{\mathcal{Y}}^{\mathsf{T}}\hat{\mathcal{Y}}\right)\}\}$
                 \State $\{\hat{\mathscr{Y}}_{i}\}_{i=(j-1)}^{nj} \gets \hat{\mathcal{Y}}(t)$
                 \State $\{\hat{\mathscr{G}}_{fui}\}_{i=(j-1)}^{nj} \gets \hat{\mathcal{G}}_{fu}(t)$
                 \State $\{\hat{\mathscr{X}}_i\}_{i=(j-1)}^{nj}\gets \hat{x}(t)-\hat{x}(t-T)$
                 \If{
                 $\mathcal{G}$ is not full}
                  \State add the data points to $\mathcal{G}$
                 \Else
                  \State add the data points to $\mathcal{G}$ if (\ref{eq:maximCond}) holds
                 \EndIf
                 \If{$\min(\eig(\hat{\mathscr{Y}}^{\mathsf{T}}\hat{\mathscr{Y}})) \geq \xi \lambda$}
                     \If {$(t-\tau_{2}) \geq\mathcal{T}(t)$}
                      \State $\mathcal{H} \gets \mathcal{G}$, $\mathcal{G} \gets 0$, and $\tau_{2} \gets t$
                      \If{$\lambda < \min(\eig(\hat{\mathscr{Y}}^{\mathsf{T}}\hat{\mathscr{Y}}))$}
                      \State $\lambda \gets \min(\eig(\hat{\mathscr{Y}}^{\mathsf{T}}\hat{\mathscr{Y}}))$
                      \EndIf
                     \EndIf
                     \State $t_{0} \gets t$, $x_{0} \gets x(t)$, $\hat{\theta}_{0} \gets \hat{\theta}(t)$
                     \State $I_{Y,0} \gets I_{Y}(t)$, $\hat{I}_{fg_{u,0}} \gets \hat{I}_{fg_{u}}(t)$
                \EndIf
            \EndIf
            \Else
            \State no event, keep on integrating the DDEs
        \EndIf
       \State $\tau_{1} \gets t$  \Comment{Set this even if a new event is not detected}
    \EndIf
     \State no event, keep on integrating the DDEs
\EndWhile
 \end{algorithmic} \label{algo:parameterEstimatorAlgo}
 \end{algorithm}

\section{Stability Analysis}\label{section:stabilityAnalysis}
In this section, stability analysis of the joint state and parameter estimation architecture will be carried out using Lyapunov methods. The following Theorem establishes local uniformly ultimately boundedness of the state estimation errors.
\begin{theorem}\label{thm:stateErrorConvergence}
    Provided Assumption \ref{ass:jacobianbounds} holds, there exists a constant symmetric positive definite matrix, $P$, and four observer gains, $l_{1}$, $l_{2}$, $l_{3}$ and $L$, that satisfy the matrix inequality,
   \begin{equation}\label{lmi}
\begin{bmatrix} 
    (A-LC)^{\mathsf{T}}P+P(A-LC) + 2 \alpha P &P-{J_{21}}^{\mathsf{T}} \\P-J_{21}  & -J_{22}
    \end{bmatrix} < 0, \end{equation}
        where $J_{21} \coloneqq (J_{y})_{21}(\mathbb{I}_{n}-l_{1}C) + (J_{f})_{21}(\mathbb{I}_{n}-l_{2}C)+(J_{g})_{21}(\mathbb{I}_{n}-l_{3}C)$ and $J_{22} \coloneqq (J_{y})_{22} + (J_{f})_{22}+(J_{g})_{22}$, then observer error system in (\ref{aug_error}) is locally uniformly ultimately bounded.
\end{theorem}
\begin{proof}
Let $\mathcal{D}$ be an open subset of the set $ \{\tilde{x} \in\mathbb{R}^n: x,\hat{x} \in \mathcal{C}\}$ and consider the continuously differentiable candidate Lyapunov function, $W: \mathcal{D} \to  \mathbb{R}$ defined as
\begin{align}
     \label{eq:lyapunov function}
    W\left(\tilde{x}\right) = {\tilde{x}}^{\mathsf{T}}P\tilde{x},
\end{align}
which satisfies the inequality
\begin{equation}\label{eq:VeBound}
    \lambda_{\min}(P)\|\tilde{x}\|^2 \leq W\left(\tilde{x}\right) \leq \lambda_{\max}(P)\|\tilde{x}\|^2. 
\end{equation}
Since $P$ is a constant symmetric positive definite matrix, both eigenvalues are positive. 
On the set, $\mathcal{D}$, the orbital derivative of the Lyapunov function along the trajectories of (\ref{aug_error}) can be expressed as
\begin{multline}
\dot{ W}(\tilde{x}) \coloneqq  
 \begin{bmatrix}
\tilde{x} \\\psi_{y}
\end{bmatrix}^{\mathsf{T}}\begin{bmatrix}
\left(\begin{gathered}
\left(K_{y_{1}}-\frac{LC}{3}\right)^{\mathsf{T}}P\\
+P\left(K_{y_{1}}-\frac{LC}{3}\right)
\end{gathered}\right)
 &P\\P& 0
\end{bmatrix}  \begin{bmatrix}
e\\\psi_{y}
\end{bmatrix}
+
\begin{bmatrix}
\tilde{x} \\\psi_{f}
\end{bmatrix}^{\mathsf{T}}\begin{bmatrix}
\left(\begin{gathered}
\left(K_{f_{1}}-\frac{LC}{3}\right)^{\mathsf{T}}P\\
+P\left(K_{f_{1}}-\frac{LC}{3}\right)
\end{gathered}\right)
 &P\\P& 0
\end{bmatrix}  \begin{bmatrix}
e\\\psi_{f}
\end{bmatrix} \\ +
 \begin{bmatrix}
\tilde{x} \\ \psi_{g}
\end{bmatrix}^{\mathsf{T}}\begin{bmatrix}
\left(\begin{gathered}\left(K_{g_1}-\frac{LC}{3}\right)^{\mathsf{T}}P\\
+P\left(K_{g_{1}}-\frac{LC}{3}\right) \end{gathered}\right) & P\\ P& 0
\end{bmatrix}\begin{bmatrix}
\tilde{x} \\ \psi_{g}
\end{bmatrix}
+  F_{\theta}(x, \tilde{\theta}). 
\end{multline}
 Provided the matrix inequalities
\begin{equation}\label{non_lmi_y}
\begin{bmatrix}
\left(K_{y_{1}}-\frac{1}{3}\left(LC\right)\right)^{\mathsf{T}}P+P\left(K_{y_{1}}-\frac{1}{3}\left(LC\right)\right) &P\\P& 0
\end{bmatrix} -
 \begin{bmatrix}
\mathbb{I}_{n}-l_{1}C & 0\\ 0 &\mathbb{I}_{n}
\end{bmatrix}^{\mathsf{T}}J_{y}\begin{bmatrix}
\mathbb{I}_{n}-l_{1}C & 0\\ 0 &\mathbb{I}_{n}
\end{bmatrix} < 0,
\end{equation}
\begin{equation}\label{non_lmi_f}
\begin{bmatrix}
\left(K_{f_{1}}-\frac{1}{3}\left(LC\right)\right)^{\mathsf{T}}P+P\left(K_{f_{1}}-\frac{1}{3}\left(LC\right)\right) &P\\P& 0
\end{bmatrix} -
 \begin{bmatrix}
\mathbb{I}_{n}-l_{2}C & 0\\ 0 &\mathbb{I}_{n}
\end{bmatrix}^{\mathsf{T}}J_{f}\begin{bmatrix}
\mathbb{I}_{n}-l_{2}C & 0\\ 0 &\mathbb{I}_{n}
\end{bmatrix} < 0,
\end{equation}
and
\begin{equation}
\begin{bmatrix}\label{non_lmi_g}
\left(K_{g_1}-\frac{1}{3}\left(LC\right)\right)^{\mathsf{T}}P+P\left(K_{g_1}-\frac{1}{3}\left(LC\right)\right) & P\\P& 0
\end{bmatrix} -
 \begin{bmatrix}
\mathbb{I}_{n}-l_{3}C & 0\\ 0 &\mathbb{I}_{n}
\end{bmatrix}^{\mathsf{T}}J_{g}\begin{bmatrix}
\mathbb{I}_{n}-l_{3}C & 0\\ 0 &\mathbb{I}_{n}
\end{bmatrix} < 0
\end{equation}
are satisfied for some constant $\alpha > 0$, the multiplier matrices and sector conditions formulated in (\ref{matrixCondB}), (\ref{matrixCondH}) and (\ref{matrixCondK}), the S-Procedure Lemma \cite{SCC.Boyd1994}, Assumption~\ref{ass:jacobianbounds}, and Assumption~\ref{ass:ThetaSet} can be used to guarantee that the orbital derivative is bounded as (cf. \cite{SCC.Behcet.Martin.ea2008})
\begin{equation}\label{eq:VeIneq}
    \dot{W}\left(\tilde{x}\right) \leq -\alpha W\left(\tilde{x}\right), \forall \tilde{x} \in \mathcal{D}, \|\tilde{x}\| \geq \xi > 0.
\end{equation}
where $\xi = \frac{\lambda_{\max}(P)\overline{F}} {\alpha\lambda_{\min}(P)}$ and $\max\limits_{x \in \mathcal{C}}\| F_{\theta}(x, \tilde{\theta})\|\leq \overline{F} \left\Vert\tilde{\theta}\right\Vert$ for some $\overline{F} \geq 0$. 

Invoking \cite[Thereom~4.18]{SCC.Khalil2002}, the state estimation error is locally uniformly ultimately bounded. And the ultimate bound on $\tilde{x}$ can be estimated as
\begin{equation}\label{eq:errorbound}
    \limsup_{t\to\infty} \|\tilde{x}\| \coloneqq \sqrt{\frac{\lambda_{\max}(P)}{\lambda_{\min}(P)}}\xi .
\end{equation}
\begin{remark}
    The observer design is only valid if the control input remains bounded and the system trajectories remain within the compact set $\mathcal{C}$ where the bounds on the Jacobians in (\ref{aug_jac_y}), (\ref{aug_jac_f}) and (\ref{aug_jac_g}), respectively, are valid. 
\end{remark}

\begin{remark}
    The matrix inequality in (\ref{lmi}) can be reformulated as a linear matrix inequality (LMI) using the typical variable substitution method. Indeed, substituting $L = P^{-1}R$ in \eqref{lmi}, the matrix $P$ and the observer gains $L$,  $l_{1}$, $l_{2}$ and $l_{3}$ can be obtained by solving the LMI
\begin{equation}\label{lmi2}
    \begin{bmatrix}
    A^\mathsf{T}P+PA-C^\mathsf{T}R^\mathsf{T}-RC + 2\alpha P & P-{J_{21}}^\mathsf{T} \\P-J_{21}  & -J_{22}
    \end{bmatrix} < 0
    \end{equation}
    for $P$, $R$, $l_{1}$, $l_{2}$ and $l_{3}$.
\end{remark}
\end{proof}
In order to rigorously analyze the convergence properties of the parameter estimation error, a precise definition of ``finitely informative'' and ``persistently informative'' data in the history stack is presented below.
\begin{definition}\label{defn:finitelInformative}
\cite{SCC.Self.Abudia.ea2022} The signal $(\hat{x},u)$ is called finitely informative (FI) if there exist time instances $0 \leq t_{1} < t_2 < \hdots < t_{N}$, for some finite positive integer $N$, such that the resulting history stack is full rank and persistently informative (PI) if, for any $T \geq 0$, there exist time instances $T \leq t_{1} < t_2 < \hdots < t_{N}$ such that the resulting history stack is full rank.
\end{definition}
The following theorem establishes that the parameter estimation error $\tilde{\theta}$ converges to a neighborhood of the origin if Assumption~\ref{ass:regressorRank} holds and the data are sufficiently informative, as per Definition~\ref{defn:finitelInformative}. To facilitate the analysis, given $s$ in $\mathbb{N}$, let $\mathcal{H}_{s}$ denote the history stack that is active during the time interval $I_s:=\{t \mid \rho(t) = s\}$ containing the data $\left\{(\hat{\mathscr{X}}_{si},\hat{\mathscr{Y}}_{si}, \hat{\mathscr{G}}_{fu_{si}})\right\}_{i=1,\hdots, N}$.
\begin{theorem}\label{thm:parameterconvergence}
    If the state and parameters of the system in (\ref{eq:dynamics_x}) are estimated using state and parameter estimators that satisfy the conditions of Theorem~\ref{thm:stateErrorConvergence} and Assumption~\ref{ass:regressorRank}, if the signal $(\hat{x},u)$ is FI,  if $\mathcal{H}$ is populated using Algorithm~\ref{algo:parameterEstimatorAlgo}, and if the excitation lasts long enough for two purging events (i.e. $\mathcal{H}_3$ is full rank), then the trajectories of the parameter estimation error are ultimately bounded.
\end{theorem}
\begin{proof}
To facilitate analysis, let $\mathit{\Psi}_{s}: \mathbb{R}_{\geq 0} \to \mathbb{R}^{p\times p}$ and $\mathit{Q}_{s}: \mathbb{R}_{\geq 0} \to \mathbb{R}^{n\times p}$ be defined as  $\mathit{\Psi}_{s} \coloneqq \sum_{i=1}^{N}\frac{{\hat{\mathcal{Y}}_{si}^{\mathsf{T}}\hat{\mathcal{Y}}_{si}}}{1+\kappa\|\hat{\mathcal{Y}}_{si}\|^{2}}$ and $\mathit{Q}_{s} \coloneqq \sum_{i=1}^{N}\frac{{\hat{\mathcal{Y}}_{si}^\mathsf{T}}\mathcal{E}_{si}}{1+\kappa\|\hat{\mathcal{Y}}_{si}\|^{2}}$. Using this notation, the dynamics of the parameter estimation error in (\ref{eq:parameterErrorDyn}) can be expressed as
\begin{equation}\label{modifiedObserverDyn}
    \dot{\tilde{\theta}} = -k_{\theta}\Gamma\mathit{\Psi}_{s}\tilde{\theta}-k_{\theta}\Gamma\mathit{Q}_{s},
\end{equation}
and (\ref{eq:gammaUpdateLaw}) can be expressed as 
\begin{equation}\label{eq:gammaUpdateLaw2}
    \dot{\Gamma} = \beta_{1}\Gamma - k_{\theta}\Gamma\mathit{\Psi}_{s}\Gamma.
\end{equation}
It is important to note that the functions $\mathit{\Psi}_{s}$ and $\mathit{Q}_{s}$ are piece-wise continuous. Thus, the trajectories of (\ref{modifiedObserverDyn}) are defined in the sense of Carathéodory \cite{SCC.Kamalapurkar.ea2017}. 

Let $\rho: \mathbb{R}_{\geq 0} \to \mathbb{N}$ denote a switching signal that satisfies initial condition $\rho(0) = 1$ and for any time $t$ in the domain of the signal, $\rho(t) = j + 1$, where $j$ denotes the number of times the update $\mathcal{H} \gets \mathcal{G}$ has been carried out over the time interval $0$ to $t$.

Using arguments similar to \cite[Theorem 1]{SCC.Kamalapurkar2017a}, provided the conditions of Theorem~\ref{thm:stateErrorConvergence} are satisfied, and the states and state estimation errors remain within the compact sets $\mathcal{C}$ and $\mathcal{D}$, respectively over the time interval $I_{s-1}$ in which the history stack was recorded, then using the error bound developed in Lemma~\ref{lem:ErrorTermformulation} the error terms can be bounded as
\begin{equation}\label{eq:epsilonBound}
    \|\mathcal{E}_{si}\| \leq  L_{e}\overline{e}_{s}, \forall i \in \{1,\hdots,N\}, \forall \tilde{x} \in \mathcal{D},
\end{equation}
where $\overline{e}_{s} \coloneqq \sup_{t \in I_{s-1}}\|\tilde{x}(t)\|$ and $L_{e} > 0$ is a constant.

Consider the candidate Lyapunov function $V: \Theta \times \mathbb{R}_{\geq 0} \to \mathbb{R}$ defined as, 
     \begin{equation}
         V(\tilde{\theta},t) \coloneqq \frac{1}{2}\tilde{\theta}^\mathsf{T}{\Gamma}^{-1}(t)\tilde{\theta}.
     \end{equation}
Using arguments similar to those presented in \cite[Section~4.4.2]{SCC.Ioannou.Sun1996}, provided (\ref{ass:regressorRank}) holds and $\lambda_{\min}\{{\Gamma(0)^{-1}}\} > 0$, the update law in (\ref{eq:gammaUpdateLaw}) ensures that the least squares update law satisfies
\begin{align}\label{eq:OFBADP1Gammabound}
\underline{\Gamma}\mathbb{I}_{p}\leq\Gamma\left(t\right)\leq\overline{\Gamma}\mathbb{I}_{p},\forall t\in\mathbb{R}_{\geq 0}	\end{align}
for some  $\overline{\Gamma},\underline{\Gamma}>0$ , where $\mathbb{I}_{p}$ denotes a $p \times p$ identity matrix.
Applying the bound in (\ref{eq:OFBADP1Gammabound}), the candidate Lyapunov function satisfies the following inequality
\begin{equation}
\frac{1}{2\overline{\Gamma}}\|\tilde{\theta}\|^{2}\leq V\left(\tilde{\theta},t\right)\leq\frac{1}{2\underline{\Gamma}}\|\tilde{\theta}\|^{2}, \forall t \in \mathbb{R}_{\geq 0}\label{eq:OFBADPVBound}.
\end{equation}

Using arguments similar to those presented in \cite[Theorem~4.4.1]{SCC.Ioannou.Sun1996}, the orbital derivative of $V$ can be bounded as,
\begin{equation}
    \dot{V}_{s}\left(\tilde{\theta}, t\right) \leq -\frac{1}{2}\underline{a}\|\tilde{\theta}\|^{2}+ k_{\theta}\|\tilde{\theta}\|\overline{Q}_{s}, 
\end{equation}
where $\underline{a} \coloneqq k_{\theta}\underline{c} + \frac{\beta_{1}}{\overline{\Gamma}}$, $\underline{c}$ is defined in  Assumption~\ref{ass:regressorRank} and $\overline{Q}_{s}$ is a positive constant such that $\overline{Q}_{s} \geq \|Q_{s}\|$. Using completion of squares, the orbital derivative is then bounded for all $t\in \mathbb{R}_{\geq 0}$ as
\begin{equation}
     \dot{V}_{s}\left(\tilde{\theta}, t\right) \leq -\frac{1}{4}\underline{a}\|\tilde{\theta}\|^2, \forall \|\tilde{\theta}\| \geq \rho(\|\mu\|)
 \end{equation}
 where $\rho(\|\mu\|) \coloneqq \sqrt{\frac{\overline{\Gamma}}{\underline{\Gamma}}}\left(\frac{4k_{\theta}}{\underline{a}}\right)\|\mu\|^{2}$ and $\mu \coloneqq \sqrt{\overline{Q}_{s}}$.
Hence, the conditions of \cite[Theorem~4.19]{SCC.Khalil2002} are satisfied and it can be concluded that (\ref{modifiedObserverDyn}) is input-to-state stable with state $\tilde{\theta}$ and input $\mu$.

If Algorithm~\ref{algo:parameterEstimatorAlgo} is implemented and if the signal $(\hat{x}, u)$ is FI, then there exists a time instance $T_{s}$, such that for all $t \geq T_{s}$, the history stack remains unchanged. And as a result, using \cite[Exercise~4.58]{SCC.Khalil2002}, an ultimate bound on $\tilde{\theta}$ can be estimated as
\begin{equation}\label{eq:ultimatebound}
     \limsup_{t\to\infty} \|\tilde{\theta} (t)\| \leq \overline{\theta}(T_{s}) \coloneqq \sqrt{\frac{\overline{\Gamma}}{\underline{\Gamma}}}\left(\frac{4k_{\theta}\overline{Q}(T_{s})}{\underline{a}}\right).
\end{equation} The parameter estimation error can be reduced by reducing the estimation errors corresponding to the state estimates stored in the history stack, which reduces $Q_{s}$.

The projection algorithm and Theorem \ref{thm:stateErrorConvergence} imply boundedness of all signals in the closed loop for all $t$. Furthermore, Theorem \ref{thm:stateErrorConvergence} implies that given any $ \varepsilon > 0, $ the gain $ \alpha $ can be selected large enough to ensure that $\tilde{x}$ has reached the ultimate bound before $ t = T_1 $, and that the ultimate bound is smaller than $\varepsilon$ so that $ \overline{e}_{2} \leq \varepsilon $. Since the history stack $\mathcal{H}_{3}$, which is active over the interval $I_{3}$, is recorded during the interval $I_{2}$, the bounds in (\ref{eq:epsilonBound}) can be used to show $\overline{Q}_{3} = \frac{N L_{e}\overline{e}_{2}}{2\sqrt{\kappa}} \leq \frac{N L_{e}\varepsilon}{2\sqrt{\kappa}}$. As such, if $(\hat{x},u)$ is FI with the excitation lasting long enough so that $\mathcal{H}_{3}$ is full rank, then \eqref{eq:ultimatebound} implies that $ \limsup_{t\to\infty} \|\tilde{\theta} (t)\| \leq \sqrt{\frac{\overline{\Gamma}}{\kappa\underline{\Gamma}}}\left(\frac{2k_{\theta}N L_{e}}{\underline{a}}\right)\varepsilon.$ 

\end{proof}

\section{Simulation}\label{section:simulation}

To demonstrate the performance of the developed method, a two-state dynamical system is simulated. 
\subsection{Two State Dynamical System}
Consider the dynamical system of the form 
\begin{equation}
    \dot{x} = Y(x)\theta + g(x)u, \quad y = Cx, 
\end{equation}
with states $x$ = $[(x)_1;(x)_2]$ where 
\begin{equation}
    Y(x) = \begin{bmatrix}
        (x)_{2} & 0 & 0 & 0 \\
        0 & (x)_{1} & (x)_{2} & x_{2}(\cos(2(x)_{1})+2)^{2}
    \end{bmatrix},
\end{equation}
$\theta = [(\theta)_{1};(\theta)_{2};(\theta)_{3};(\theta)_{4}]$, $g(x)=[0;\cos(2(x)_{1})+2]$ and $C = [1;0]^\mathsf{T}$.

To satisfy Assumption~\ref{ass:regressorRank}, a controller that results in a uniformly bounded system response is needed.  For the purpose of this simulation study, the controller, denoted as $u$, is chosen to be a proportional-derivative (PD) controller, represented by the equation $u = -k_{p}\left((x)_{1}-x_{d}\right)-k_{d}\left((x)_{2}-\dot{x}_{d}\right)$, where $k_{p}$ and $k_{d}$ are constants that control the proportional and derivative terms, respectively. The objective of this controller is to make the system track the trajectory $(x_d)_{1}(t)=(x_d)_{2}(t) = -\frac{1}{3}\cos(3t)-\frac{1}{2}\cos(2t)$. The initial conditions of the systems are selected as $x(0) = [2;2]$, $\hat{x}(0) = [2.5;1.5]$, $\hat{\theta}(0) = [0;0;0;0]$. The actual values of the unknown parameters in the system model are $(\theta)_{1} = 1,\ (\theta)_{2} = -1,\ {\theta}_3 = -0.5,\ {\theta}_4 = 0.5$. 

In order to satisfy the stability conditions of Theorem~\ref{thm:stateErrorConvergence}, the LMI in (\ref{lmi}) is solved using SEDUMI in YALMIP on MATLAB. The objective is to obtain the three observer gains, $L$, $l_{1}$, and $l_{2}$, as well as the symmetric positive definite matrix, P, which satisfies the LMI. The learning rate used in the LMI is $\alpha = 2$.

Data is added to the history stack $\mathcal{H}$ using the minimum eigenvalue maximization algorithm detailed in Algorithm~ \ref{algo:parameterEstimatorAlgo} with initial values given as $I_{Y,0} = 0_{2\times 4}$ and $I_{fgu,0} = 0_{2\times 1} $, $T = 2$, $t^{*} = 0.1$, $\lambda^{*} = 0$. The learning gains are selected, through
trial and error, as $N=25$, $k_{\theta} = 50$, $\beta_{1} = 0.5$, $\Gamma(0) = \diag([1,1,1,1])$, $k_{p}=[50;50]^\mathsf{T}$.
\begin{figure}[H]
        \centering
        \begin{tikzpicture}
            \begin{axis}[
                xlabel={$t$ [s]},
                ylabel={$\tilde{x}(t)$},
                legend pos = north east,
                legend style={nodes={scale=0.75, transform shape}},
                enlarge y limits=0.05,
                enlarge x limits=0,
                width=\linewidth,
                height=0.5\linewidth,
            ]
            \pgfplotsinvokeforeach{1,...,2}{
                \addplot+ [thick, mark=none] table [x index=0, y index=#1] {data/downsampledxTilde.dat};
            }
            \legend{$\tilde{x}_1$, $\tilde{x}_2$}
            \end{axis}
        \end{tikzpicture}
            \caption{Trajectory of error between the actual states and the estimated states}
		\label{fig:state_error_sim1}
\end{figure}
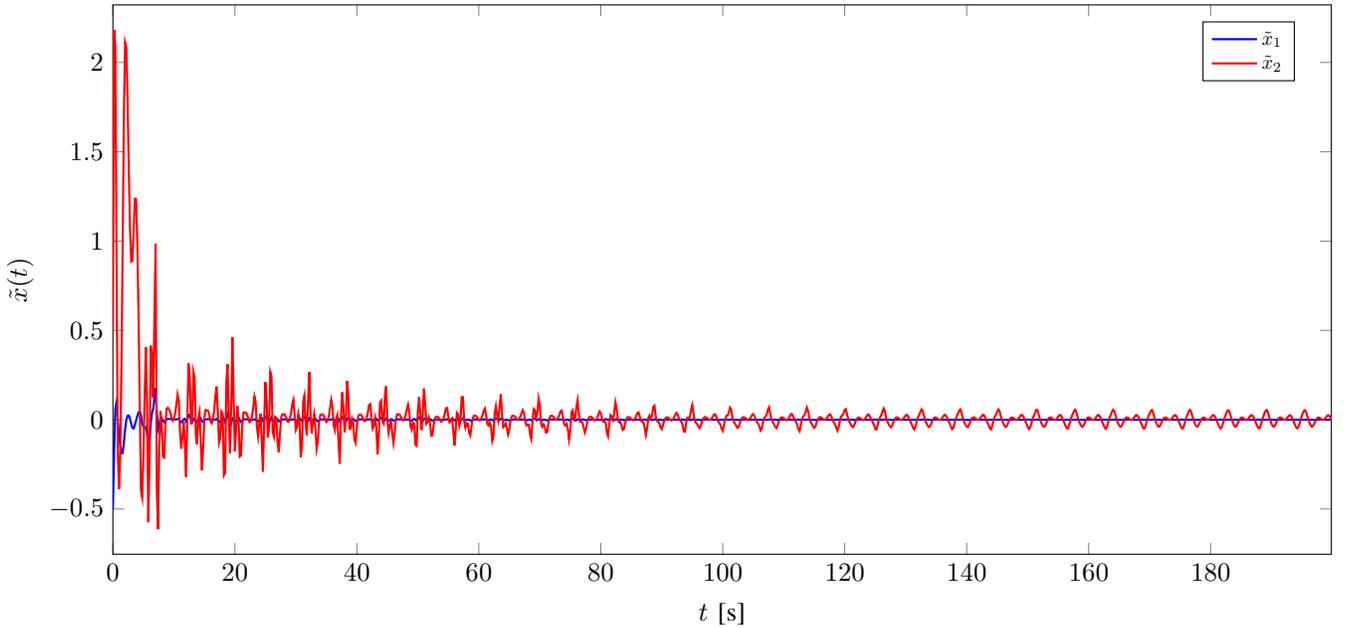
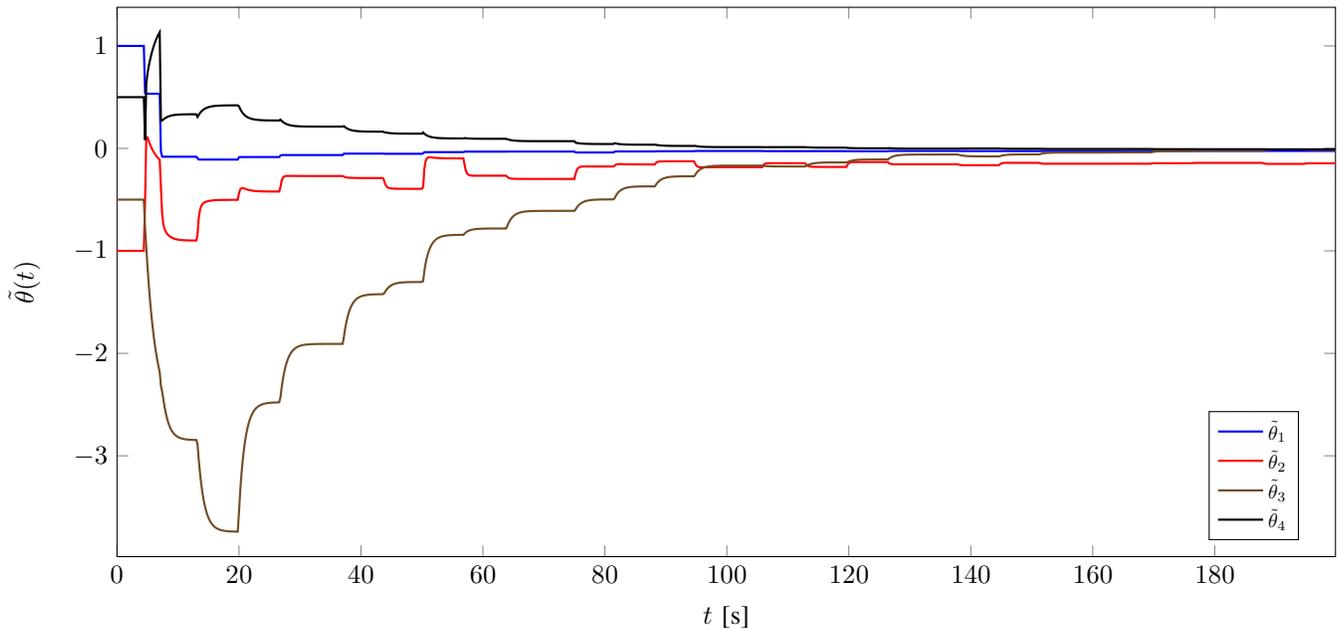
\begin{figure}
        \centering
        \begin{tikzpicture}
            \begin{axis}[
                xlabel={$t$ [s]},
                ylabel={$\tilde{\theta}(t)$},
                legend pos = south east,
                legend style={nodes={scale=0.75, transform shape}},
                enlarge y limits=0.05,
                enlarge x limits=0,
                width=\linewidth,
                height=0.5\linewidth,
            ]
            \pgfplotsinvokeforeach{1,...,4}{
                \addplot+ [thick, mark=none] table [x index=0, y index=#1] {data/downsampledThetaTilde.dat};
            }
            \legend{$\tilde{\theta}_1$, $\tilde{\theta}_2$, $\tilde{\theta}_3$, $\tilde{\theta}_4$}
            \end{axis}
        \end{tikzpicture}
		\caption{Trajectory of error between the actual parameters and the estimated parameters}
		\label{fig:parameter_error_sim1}
\end{figure}

\subsection{Results and Discussion} \label{ Result_sim1}
Figure~\ref{fig:state_error_sim1} and figure~\ref{fig:parameter_error_sim1} demonstrate that the developed state and parameter estimators are effective in driving the trajectories of state estimation errors and parameter estimation errors to the origin, respectively. This result demonstrates the effectiveness of the developed method and validates the theoretical results in Section~\ref{section:stabilityAnalysis}. The value of the constant symmetric positive definite matrix $P$ is given as $P =[2.3886,-0.1840;-0.1840,0.0270]$, and the value of the observer gain was found to be $L = [10.0671;
  103.167]$.  
  
  Algorithm~\ref{algo:parameterEstimatorAlgo} was implemented with MATLAB's DDE solver. The event-based implementation was employed to avoid unexpected errors that may arise from numerically integrating discontinuous differential equations with variable step-size solvers due to the removal and addition of data to the history stack. Given the delay $T$, the solver kept track of the solution of the DDEs at time $t$ and time $t-T$. Then an event function was passed as an optional argument to the DDE solver and set up to stop integration as soon as an event is detected as described by Algorithm~\ref{algo:parameterEstimatorAlgo}.

\section{Conclusion}\label{section:conclusion}

In this paper, an online joint state and parameter estimation scheme is developed for nonlinear systems is proposed using a multiplier matrix observer design and a novel event-based implementation of concurrent learning adaptive update laws. Convergence properties of the developed method are analyzed using Lyapunov methods and validated through simulation, demonstrating local uniformly ultimately boundedness of the state estimation errors and input-to-state stability of parameter estimation errors under a finite informativity condition. Additionally, a persistent informativity condition guarantees convergence of the parameter estimation errors to a neighbourhood of the origin.

To improve the applicability of the observer design to a wider range of nonlinear systems and to allow for relaxed LMI conditions, future work will involve developing a methodology for simultaneous state and parameter estimation via exact Takagi-Sugeno tensor-product models or polynomial rewriting of the error system, as formulated in \cite{SCC.Quintana.Bernal.ea2021,SCC.Guerra.Bernal.ea2018}. Additionally, the current LMI architecture can be augmented with techniques such as \cite{SCC.Bengt.Richard.ea2012} that uses a delta operator formulation to address the rank deficiency in poorly conditioned LMIs.

\small
 \bibliographystyle{IEEETrans.bst}
\bibliography{scc, sccmaster,scctemp}
 
\end{document}